\documentclass[conference,a4paper]{IEEEtran}

\usepackage{amsmath}
\usepackage{amssymb}
\usepackage{amsfonts}
\usepackage{amscd}
\usepackage{amsthm}
\usepackage[noadjust]{cite}

\usepackage{mathrsfs}
\usepackage{bbm}
\usepackage{amsbsy}

\usepackage{paralist}

\usepackage{xfrac}

\newtheorem{theorem}{Theorem}
\newtheorem{lemma}{Lemma}

\newtheorem{corollary}{Corollary}

\theoremstyle{definition}
\newtheorem{example}{Example}

\newtheorem{remark}{Remark}

\makeatletter
\renewcommand*\env@matrix[1][c]{\hskip -\arraycolsep
  \let\@ifnextchar\new@ifnextchar
  \array{*\c@MaxMatrixCols #1}}
\makeatother
\usepackage{graphicx}

\usepackage{etoolbox}
\AtEndEnvironment{example}{\null\hfill\IEEEQED}%

\newcommand{\Fb}{\mathbbmss{F}} 
\newcommand{\tth}{{\text{th}}} 
\newcommand{\rx}{{\sf Rx}}
\newcommand{\tr}{{\sf T}} 
\newcommand{\ravg}{{r_{\rm avg}}}
\newcommand{\Db}{{\mathcal{D}}}
\newcommand{\Kb}{{\mathcal{K}}}
\newcommand{\Sc}{{\mathcal{S}}}
\newcommand{\Mc}{{\mathcal{M}}}
\newcommand{\supp}{{\sf supp}}
\newcommand{\spp}{{\sf span}}
\newcommand{\colsp}{{\mathcal{C}}}
\newcommand{\nullsp}{{\mathcal{N}}}
\newcommand{\rank}{{\rm rank}}
\newcommand{\minrank}{{{\rm minrk}_q}}
\newcommand{\Vin}{{V_{\sf in}}}
\newcommand{\Vout}{{V_{\sf out}}}

\title{Locality in Index Coding for Large Min-Rank}
\author{Lakshmi Natarajan, Hoang Dau, Prasad Krishnan and V.\ Lalitha}

\begin{document}

\maketitle

\begin{abstract}
\boldmath
An index code is said to be \emph{locally decodable} if each receiver can decode its demand using its side information and by querying only a subset of the transmitted codeword symbols instead of observing the entire codeword. 
Local decodability can be a beneficial feature in some communication scenarios, such as when the receivers can afford to listen to only a part of the transmissions because of limited availability of power.
The \emph{locality} of an index code is the ratio of the maximum number of codeword symbols queried by a receiver to the message length. 
In this paper we analyze the optimum locality of linear codes for the family of index coding problems whose min-rank is one less than the number of receivers in the network. 
We first derive the optimal trade-off between the index coding rate and locality with vector linear coding when the side information graph is a directed cycle. 
We then provide the optimal trade-off achieved by scalar linear coding for a larger family of problems, viz., problems where the min-rank is only one less than the number of receivers.
While the arguments used for achievability are based on known coding techniques, the converse arguments rely on new results on the structure of locally decodable index codes.
 
\end{abstract}

\section{Introduction} \label{sec:introduction}

Index coding~\cite{YBJK_IEEE_IT_11} is a central problem in network coding theory, because of its applications, such as in video-on-demand and daily newspaper delivery~\cite{BiK_INFOCOM_98}, and because of its strong relation to other coding theoretic problems, such as network coding~\cite{RSG_IEEE_IT_10,EEL_IT_15}, coded caching~\cite{MaN_IT_14}, codes for distributed data storage~\cite{Maz_ISIT_14,ShD_ISIT_14} etc.
The index coding problem is to design a code for a broadcast channel where each receiver has prior side information of a subset of messages being transmitted. 
The objective is to minimize the number of uses of the broadcast channel, or equivalently, the \emph{broadcast rate}. 

\let\thefootnote\relax\footnotetext{Dr.\ Natarajan is with the Department of Electrical Engineering, Indian Institute of Technology Hyderabad, email: lakshminatarajan@iith.ac.in. 

Dr.\ Dau is with the Department of Electrical and Computer Systems Engineering, Monash University, Melbourne, Australia, email: hoang.dau@monash.edu.

Dr.\ Krishnan and Dr.\ Lalitha are with the Signal Processing \& Communications Research Center, International Institute of Information Technology Hyderabad, India, email:\{prasad.krishnan,\,lalitha.v\}@iiit.ac.in.
}

Conventional index coding solutions in the literature require each receiver to observe the entire transmitted codeword in order to decode its demand. 
If the network involves a large number of receivers the number of transmissions that each receiver has to observe could be significantly larger than the size of the message itself.
Thus conventional index coding solutions could be unfavorable in certain applications, such as when the power available at the wireless receivers is limited and they can not afford to listen to radio transmissions for an extended period of time. 
In such scenarios, it is desirable to use \emph{locally decodable index codes}~\cite{HaL_ISIT_12}, which provide reductions in broadcast rate while requiring the receivers to query only a part of the transmitted codeword.
The \emph{locality} of an index code is the ratio of the number of codeword symbols queried by a receiver to the number of message symbols it demands~\cite{NKL_ISIT_18}.
The objective of locally decodable index coding is to minimize both broadcast rate and locality simultaneously, and achieve the optimal trade-off between these two parameters.

To the best of our knowledge, the idea of local decodability in index coding was introduced in~\cite{HaL_ISIT_12} where the broadcast rates of random index coding problems, modeled as random graphs, were analyzed under a locality requirement. 
The results in~\cite{HaL_ISIT_12} and~\cite{NKL_ISIT_18} characterize the optimal broadcast rate of an arbitrary index coding problem when locality is set to the minimum possible value, which is unity. Constructions of index codes with locality greater than one were given in~\cite{NKL_ISIT_18}.
Locally decodable index codes were shown to be related to privacy in index coding in~\cite{KSCF_arXiv_18} and studied under the terminology `$k$-limited access schemes'. The authors of~\cite{KSCF_arXiv_18} provide constructions that modify any given binary scalar linear index code into a locally decodable scalar linear code at the cost of increased broadcast rate.

Determining the optimal trade-off between broadcast rate and locality of a given index coding problem is yet to be addressed in the literature. 
This will not only involve designing good achievability schemes but also formulating tight lower bounds on locality and rate. 

In this paper we consider linear index codes for the family of index coding problems whose \emph{min-rank} is one less than the number of receivers in the problem.
We consider the problems where the side information graph is a directed cycle, and derive the optimal trade-off between rate and locality when vector linear codes are used. We also analyse the dependence of locality on the length of the vector messages in vector linear index codes for directed cycles (Section~\ref{sec:directed_cycles}).
We then consider the larger class of index coding problems, viz., problems whose min-rank is one less than the number of receivers, and provide the exact trade-off between rate and locality when scalar linear codes are used (Section~\ref{sec:minrank_N_minus_1}).
Note that directed cycles are a subset of this larger class of problems.
The derivation of the optimal trade-off provided in Sections~\ref{sec:directed_cycles} and~\ref{sec:minrank_N_minus_1} rely on new results on the structural properties of locally decodable index codes given in Section~\ref{sec:structure}.

The achievability schemes used in this paper are based on known index coding techniques for directed cycles. 
However, one of the main contributions of this paper is the development of new tools to derive good lower bounds on rate and locality that are vital in proving the optimality of these schemes. 

\noindent
\emph{Notation:} For any positive integer $N$, we will denote the set $\{1,\dots,N\}$ by $[N]$. Matrices and column vectors are denoted by bold upper and lower case letters, respectively, such as $\pmb{A}$ and $\pmb{x}$. The finite field of size $q$ is denoted as $\Fb_q$. The subspace spanned by vectors $\pmb{u}_1,\dots,\pmb{u}_N$ is denoted by $\spp(\pmb{u}_1,\dots,\pmb{u}_N)$. 
The column space of a matrix $\pmb{A}$ is denoted as $\colsp(A)$ and the null space of $\pmb{A}$ is $\nullsp(\pmb{A})=\{\pmb{x}|\pmb{Ax}=\pmb{0}\}$.
The support set of a vector $\pmb{x}$ is denoted as $\supp(\pmb{x})$.

\section{System Model \& Preliminaries} \label{sec:system_model}

We consider index coding for a broadcast channel consisting of $N$ receivers $\rx_1,\dots,\rx_N$. The transmitter holds $N$ messages $\pmb{x}_1,\dots,\pmb{x}_N$ where the $i^\tth$ message is demanded by $\rx_i$, and the messages $\pmb{x}_j$, $j \in K_i$, are known at this receiver as side information, where $K_i \subset [N]$. 
The side information graph $G=(\mathcal{V},\mathcal{E})$ is the directed graph with vertex set $\mathcal{V} = [N]$ and edge set $\mathcal{E}=\{(i,j)\,|\, i \in [N], j \in K_i\}$, and it completely specifies the index coding problem.

We are interested in linear index codes, and hence, we will assume that each message $\pmb{x}_i$ is a vector of length $M$ over a finite field $\Fb_q$. 
Note that for scalar linear index coding problems the message length $M=1$.
The $M$ components of the $i^\tth$ message vector $\pmb{x}_i$ are denoted as $x_{i,1},\dots,x_{i,M}$.
Encoding is performed by first concatenating the $N$ messages into $\pmb{x} = (\pmb{x}_1^\tr,\dots,\pmb{x}_N^\tr)^\tr \in \Fb_q^{MN}$ and multiplying this vector with an encoding matrix $\pmb{L} \in \Fb_q^{MN \times \ell}$ to generate a length $\ell$ the codeword $\pmb{c}^\tr = \pmb{x}^\tr \pmb{L}$.
Note that the $MN$ components of the concatenated vector $\pmb{x}=(x_1,\dots,x_{MN})^\tr$ and the components of the individual message vectors are related as $x_{(i-1)M+m} = x_{i,m}$ for $i \in [N]$ and $m \in [M]$.
The code length corresponding to the encoding matrix $\pmb{L}$ is $\ell$ and the broadcast rate is $\beta = \frac{\ell}{M}$.

Unlike the conventional index coding scenario where each receiver is required to query or download the entire codeword $\pmb{c}$, we allow the receivers to query only a part of the transmitted codeword in order to decode their demands.
Index codes that satisfy this property are called \emph{locally decodable}~\cite{HaL_ISIT_12}.
We will assume that $\rx_i$ queries the subvector $\pmb{c}_{R_i}=(c_j, j \in R_i)$, where $R_i \subseteq [\ell]$ is chosen in such a way that $\rx_i$ can decode $\pmb{x}_i$ using $\pmb{c}_{R_i}$ and the available side information $\pmb{x}_j$, $j \in K_i$.
The \emph{locality} of $\rx_i$ is $r_i = \frac{|R_i|}{M}$. Since $\rx_i$ demands message symbols $x_{i,1},\dots,x_{i,M}$, it needs to query at least $M$ components of the codeword to be able to decode them, and hence $r_i \geq 1$. The \emph{overall locality} or simply the \emph{locality} of the index code is $r = \max_{i \in [N]} r_i$, and the \emph{average locality} is 
\begin{equation} \label{eq:ravg}
\ravg = \sum_{i \in [N]}\frac{r_i}{N} = \sum_{i \in [N]} \frac{|R_i|}{MN}.
\end{equation} 
Observe that the average locality is upper bounded by overall locality
$\ravg \leq r$.

\begin{example}[\emph{A simple scalar linear code for directed cycles}] \label{ex:simple_code_cycle}
Let the message length $M=1$ and let $G$ be a directed cycle of length $N$, i.e., $K_i=\{i+1\}$ for $i \in [N-1]$ and $K_N=\{1\}$. Consider the $N \times (N-1)$ encoder matrix
\begin{equation*}
\pmb{L} = \begin{bmatrix} 
          1 & 1 & \cdots & 1 \\
          1 & 0 & \cdots & 0 \\
          0 & 1 & \cdots & 0 \\
          \vdots &  &  & \vdots \\
          0 & 0 & \cdots & 1
          \end{bmatrix},
\end{equation*} 
that generates the codeword 
\begin{equation*}
\pmb{c} = (c_1,\dots,c_{N-1}) = \pmb{x}^\tr\pmb{L} = (x_1+x_2,x_1+x_3,\dots,x_1+x_N).
\end{equation*} 
$\rx_1$ and $\rx_N$ can decode their demands by querying $c_1$ and $c_{N-1}$, respectively, hence, $|R_1|=|R_N|=1$. For $1 < i < N$, receiver $\rx_i$ queries the symbols $c_{i-1}=x_1+x_i$ and $c_{i}=x_1 + x_{i+1}$, and uses its side information $x_{i+1}$ to compute $c_{i-1} - c_i - x_{i+1}$, which equals its demand $x_i$. Hence, $|R_i| = 2$ for $1 < i < N$. Since $M=1$, the locality of each receiver $r_i=1$ if $i=1$ or $N$, and $r_i=2$ otherwise. The overall locality $r=2$ and the average locality $\ravg=2(N-1)/N$.

Since the graph $G$ is symmetric, for any choice of $i \in [N-1]$, the above coding scheme can be modified by an appropriate permutation of the rows of $\pmb{L}$ to allow receiver localities $r_i=1$ and $r_{i+1}=1$ and locality $r_j=2$ at all other receivers $j \neq i,i+1$. 
\end{example}

We would like to characterize the trade-off between the broadcast rate $\beta$ and the locality $r$ of linear index codes over $\Fb_q$ for a given index coding problem $G$. We define the optimum locality-rate trade-off among all linear index codes for $G$ over $\Fb_q$ as
\begin{equation*}
 \beta^*_{G,q}(r) = \inf \left\{ \beta \, | \, \exists \text{ a linear code of rate } \beta, \text{ locality} \leq r  \right\},
\end{equation*} 
where the infimum considers linear index codes over all possible message lengths $M \geq 1$ for the index coding problem $G$.


We would like to view the vector linear index coding problem involving $N$ vector messages $\pmb{x}_1,\dots,\pmb{x}_N \in \Fb_q^M$ as a scalar linear problem defined over $MN$ scalar messages $x_1,\dots,x_{MN} \in \Fb_q$. In this case the $i^\tth$  receiver demands $M$ scalar messages $x_{(i-1)M+1},x_{(i-1)M+2},\dots,x_{iM}$, which correspond to the $M$ components of the vector $\pmb{x}_i$. This set of demands of $\rx_i$ is represented by the index set $\Db_i = \{(i-1)M+m \,|\, m \in [M] \}$. The scalar symbols available as side information at $\rx_i$ correspond to the index set 
\begin{equation*}
 \Kb_i = \left\{ (j-1)M + m \, \vert \, j \in K_i, m \in [M] \right\}.
\end{equation*} 
Thus the vector linear problem corresponding to the side information graph $G$ with message length $M$, is equivalent to a scalar linear problem with $N$ receivers and $MN$ messages, where the index set of demands of $\rx_i$ is $\Db_i$ and the set corresponding to the side information at $\rx_i$ is $\Kb_i$. 
Note that $\Db_i$ and $\Kb_i$ are subsets of $[MN]$.

\section{Structure of Locally Decodable Index Codes} \label{sec:structure}

We will first derive some properties of locally decodable (vector) linear index codes (codes for any message length $M \geq 1$) in Section~\ref{sec:linear_structure} and then analyze scalar linear codes (message length $M=1$) specifically in Section~\ref{sec:scalar_structure}. 
These properties will be useful in deriving the rate-locality trade-off results presented in later sections. 

Following the notation from~\cite{DSC_IT_12}, for a vector $\pmb{u} \in \Fb_q^{MN}$ and set $E \subset [MN]$, we write $\pmb{u} \lhd E$ to denote $\supp(\pmb{u}) \subseteq E$.
Let us denote the columns of the encoder matrix $\pmb{L}$ as $\pmb{L}_1,\dots,\pmb{L}_\ell \in \Fb_q^{MN}$. Then the $k^\tth$ symbol of the codeword is $c_k = \pmb{x}^\tr\pmb{L}_k$.
Note that the $i^\tth$ receiver queries the subvector $\pmb{c}_{R_i}=(c_k,k \in R_i)$, and utilizes the side information $\pmb{x}_{\Kb_i}=(x_j, j \in \Kb_i)$, to decode the demand $\pmb{x}_{\Db_i}=(x_j, j \in \Db_i)$. 

\subsection{Locally decodable linear index codes} \label{sec:linear_structure}

We observe that the proofs of Lemmas~3.1 and~4.3 and Corollary~4.4 of~\cite{DSC_IT_12} can be directly adapted to the scenario of locally decodable index codes, immediately yielding the following constraints on the encoder matrix.
Let $\pmb{e}_1,\dots,\pmb{e}_{MN}$ be the standard basis of $\Fb_q^{MN}$.

\begin{theorem} \label{thm:design_criterion}
For each $i \in [N]$, let $\rx_i$ query the subvector $\pmb{c}_{R_i}$ of the codeword $\pmb{c}=\pmb{x}^\tr\pmb{L}$ and have the side information $\pmb{x}_{\Kb_i}$. Then $\rx_i$ can decode its demand $\pmb{x}_{\Db_i}$ if and only if for each $j \in \Db_i$ there exists a $\pmb{u}_j \in \Fb_q^{MN}$ that $\pmb{u}_j \lhd \Kb_i$ and $\pmb{u}_j + \pmb{e}_j \in \spp(\pmb{L}_k, k \in R_i)$.
\end{theorem}

We will say that $\pmb{L} \in \Fb_q^{MN \times \ell}$ is a valid encoder matrix corresponding to the queries $R_1,\dots,R_N \subseteq [\ell]$ if it satisfies the criterion stated in Theorem~\ref{thm:design_criterion} for decodability at all the receivers.

Observe that among the component symbols in the codeword $\pmb{c}=(c_1,\dots,c_\ell)^\tr$, some are queried exactly once, i.e., queried by a single receiver, and the other symbols are queried by multiple receivers in the network. Let 
$\Sc_i = R_i \setminus \left(R_1 \cup \cdots \cup R_{i-1} \cup R_{i+1} \cdots \cup R_N \right)$
denote the index set of coded symbols that are queried only by $\rx_i$. Also, let
$\Mc_i = R_i \cap \left(R_1 \cup \cdots \cup R_{i-1} \cup R_{i+1} \cdots \cup R_N \right)$
denote the index set of coded symbols that are queried by $\rx_i$ and at least one other receiver. Note that $R_i = \Sc_i \cup \Mc_i$ for each $i \in [N]$.

The following result shows that certain entries of the encoder matrix can be set to be equal to zero without affecting the locality or rate of the index code.

\begin{theorem} \label{thm:zeroes}
Let $\pmb{L} \in \Fb_q^{MN \times \ell}$ be a valid encoding matrix corresponding to the receivers' queries $R_1,\dots,R_N$. Then there exists a valid encoding matrix $\pmb{L}' \in \Fb_q^{MN \times \ell}$ for the queries $R_1,\dots,R_N$ such that for each $i \in [N]$
\begin{equation*}
 \pmb{L}'_k \lhd \Db_i \text{ for all } k \in \Sc_i.
\end{equation*} 
\end{theorem}
\begin{proof}
See Appendix~\ref{app:proof:thm:zeroes}.
\end{proof}

The new index code guaranteed by Theorem~\ref{thm:zeroes} employs the same code length and the same set of queries as the given index code. Hence, the broadcast rate $\beta$, overall locality $r$ and the average locality $\ravg$ of the new code are identical to those of the given index code. Additionally, the new code guarantees that for any $i \in [N]$ any codeword symbol $c_k$, $k \in \Sc_i$, queried only by $\rx_i$, can be expressed as a linear combination of the demands of $\rx_i$. 
Since any valid encoder matrix can be modified to satisfy this property using Theorem~\ref{thm:zeroes}, in the sequel, without loss of generality, we will only consider encoder matrices $\pmb{L}$ that satisfy
\begin{equation} \label{eq:zeroes}
 \pmb{L}_k \lhd \Db_i \text{ for all } k \in \Sc_i \text{ and } i \in [N].
\end{equation} 

Further, without loss of generality, we will assume that for each $i \in [N]$, the vectors $\pmb{L}_k$, $k \in R_i$, are linearly independent. If this is not the case, then at least one of the codeword symbols $c_k$ queried by $\rx_i$ is a linear combination of the other queried symbols $\pmb{c}_{R_i \setminus \{k\}}$. Reducing the index set of the queries of $\rx_i$ from $R_i$ to $R_i \setminus \{k\}$ does not affect decodability at $\rx_i$ since $c_k$ can be reconstructed from $\pmb{c}_{R_i \setminus \{k\}}$. 
Note that this reduction in the queries does not increase the value of either the overall locality $r$ or the average locality $\ravg$ of the index code.
This process can be repeated till the columns of $\pmb{L}$ corresponding to the queries of each of the receivers are linearly independent. 
Finally, any codeword symbol that is not queried by any of the receivers can be removed from the transmission since this symbol will not be used for decoding.

Let us denote the index set of codeword symbols that are queried exactly once by
\begin{equation}
\Sc = \Sc_1 \cup \cdots \cup \Sc_N.
\end{equation} 
Note that $\Sc_i \cap \Sc_j = \phi$ for any $i \neq j$. Hence, $|\Sc|=\sum_{i=1}^{N}|\Sc_i|$.
Let 
\begin{equation}
\Mc = \Mc_1 \cup \cdots \cup \Mc_N 
\end{equation} 
denote the index set corresponding to the codeword symbols that have been queried by more than one receiver. Observe that $\Sc \cap \Mc = \phi$ and $\Sc \cup \Mc = [\ell]$.

\begin{lemma} \label{lem:single_queries}
For any valid index code for message length $M$, number of receivers $N$, rate $\beta$ and average locality $\ravg$, 
\begin{equation*}
 |\Sc| \geq M(2\beta - N\ravg),
\end{equation*} 
where $|\Sc|$ is the number of codeword symbols that have been queried exactly once.
\end{lemma}
\begin{proof}
We will count the total number of queries made by all the receivers in two different ways and relate these expressions to arrive at the statement of this lemma. 

The number of queries made by $\rx_i$ is $|R_i|$. Hence, the total number of queries made by all the receivers is $\sum_{i \in [N]}|R_i|$. From~\eqref{eq:ravg} this is equal to $MN\ravg$. The number of times a codeword symbol $c_k$ is queried is equal to $1$ if $k \in \Sc$, and is at least $2$ if $k \in \Mc$. Thus the total number of queries satisfies
\begin{align*}
 MN\ravg &\geq \sum_{k \in \Sc} 1 \, + \, \sum_{k \in \Mc} 2= |\Sc| + 2|\Mc| \\ 
&= |\Sc| + 2(\ell - |\Sc|) = 2\ell - |\Sc|,
\end{align*} 
where we have used the fact $|\Sc| + |\Mc|=\ell$. Substituting $\ell = M\beta$ in the above inequality, we arrive at $MN\ravg \geq 2M\beta - |\Sc|$ thereby proving the lemma.
\end{proof}

\subsection{Scalar linear index codes with local decodability} \label{sec:scalar_structure}

We will now derive a few results that hold for the case $M=1$. Note that, in this case $\Db_i=\{i\}$ and $\Kb_i=K_i$ for all $i \in [N]$. The locality of each receiver $r_i=|R_i|$ is an integer, and so is the overall locality $r=\max_i r_i$. 

A matrix $\pmb{A} \in \Fb_q^{N \times N}$ \emph{fits} $G=(\mathcal{V},\mathcal{E})$ if the diagonal elements of $\pmb{A}$ are all equal to $1$, and the $(j,i)^\tth$ entry of $\pmb{A}$ is zero if $j \notin K_i$, i.e., if $(i,j) \notin \mathcal{E}$. The \emph{minrank} of $G$ over $\Fb_q$ is the minimum among the ranks of all possible matrices $\pmb{A} \in \Fb_q^{N \times N}$ that fit $G$, and is denoted as $\minrank(G)$. It is known that the smallest possible scalar linear index coding rate is equal to $\minrank(G)$~\cite{YBJK_IEEE_IT_11,DSC_IT_12}. We also know from~\cite{YBJK_IEEE_IT_11,DSC_IT_12} that a matrix $\pmb{L}$ is a valid encoder matrix for $G$ if and only if for each receiver $i \in [N]$, there exists a vector $\pmb{u}_i \in \Fb_q^N$ such that $\pmb{u}_i \lhd K_i$ and $\pmb{u}_i + \pmb{e}_i \in \colsp(\pmb{L})$, where $\colsp$ denotes the column span of a matrix. 
If $\pmb{L}$ is a valid encoder matrix, stacking these vectors we obtain the $N \times N$ matrix 
\begin{equation*}
\pmb{A} = \left[ \pmb{u}_1 + \pmb{e}_1~~\pmb{u}_2+\pmb{e}_2~\cdots~\pmb{u}_N+\pmb{e}_N  \right].
\end{equation*} 
Notice that $\pmb{A}$ fits $G$ and $\colsp(\pmb{A}) \subseteq \colsp(\pmb{L})$. We will say that $\pmb{A}$ is a \emph{fitting matrix corresponding to the encoder matrix} $\pmb{L}$.

Suppose $\pmb{L} \in \Fb_q^{N \times \ell}$ is a valid scalar linear encoder and the queries of the $N$ receivers are $R_1,\dots,R_N \subseteq [\ell]$. 
From Theorem~\ref{thm:design_criterion}, for each $i \in [N]$, there exists a vector $\pmb{u}_i \lhd K_i$ such that $\pmb{u}_i + \pmb{e}_i \in \spp(\pmb{L}_k, k \in R_i)$.
Thus, there exist scalars $\alpha_{i,k}$, $k \in R_i$, such that $\pmb{u}_i+\pmb{e}_i = \sum_{k \in R_i} \alpha_{i,k}\pmb{L}_k$. 
The $i^\tth$ receiver decodes its demand by computing $\sum_{k \in R_i}\alpha_{i,k}c_k - \pmb{x}^\tr \pmb{u}_i$, which is equal to
\begin{align*}
\sum_{k \in R_i} \alpha_{i,k} \pmb{x}^\tr\pmb{L}_k - \pmb{x}^\tr\pmb{u}_i = \pmb{x}^\tr(\pmb{u}_i + \pmb{e}_i) - \pmb{x}^\tr\pmb{u}_i = x_i.
\end{align*} 
Notice that the receiver can compute $\pmb{x}^\tr\pmb{u}_i$ using its side information since $\supp(\pmb{u}_i) \subseteq K_i$.
We will assume that each scalar $\alpha_{i,k}$ is non-zero since if $\alpha_{i,k}=0$ the receiver does not need to query the coded symbol $c_k$.
Finally, notice that stacking the vectors $\pmb{u}_i + \pmb{e}_i$, $i \in [N]$, we obtain a fitting matrix $\pmb{A}$ corresponding to $\pmb{L}$.
  
For certain choices of $S \subseteq [N]$, we will now relate the sizes of $R_i$, $i \in S$, and their union $\cup_{i \in S}R_i$. Let $\nullsp(\pmb{A})$ denote the null space of $\pmb{A}$.

\begin{lemma} \label{lem:fitting_support}
Let $\pmb{A}$ be any fitting matrix corresponding to a valid scalar linear encoder $\pmb{L}$, and $S \subseteq [N]$ be such that $S$ is the support of a non-zero vector in $\nullsp(\pmb{A})$ and the vectors $\pmb{L}_k$, $k \in \cup_{i \in S}R_i$, are linearly independent. Then
\begin{equation*}
\sum_{i \in S}|R_i| ~\geq~ 2\,\left| \bigcup_{i \in S} R_i \right|.
\end{equation*} 
\end{lemma}
\begin{proof}
Denote the columns of $\pmb{A}$ by $\pmb{A}_1,\dots,\pmb{A}_N$. Let $\pmb{z} \in \nullsp(\pmb{A}) \setminus \{\pmb{0}\}$ be such that $S = \supp(\pmb{z})$. Since $\pmb{Az}=\pmb{0}$, we have $\sum_{i \in S}z_i\pmb{A}_i=\pmb{0}$, where the components $z_i$, $i \in S$, of the vector $\pmb{z}$ are non-zero.
Notice that there exist non-zero scalars $\alpha_{i,k}$ such that $\pmb{A}_i = \sum_{k \in R_i} \alpha_{i,k} \pmb{L}_k$. Hence, we have
\begin{align*}
\pmb{0} = \sum_{i \in S}z_i\pmb{A}_i = \sum_{i \in S}\sum_{k \in R_i} z_i\alpha_{i,k} \pmb{L}_k.
\end{align*} 
All the scalars $z_i\alpha_{i,k}$ in the above linear combination are non-zero, and the set of vectors $\pmb{L}_k$, $k \in \cup_{i \in S}R_i$, appearing in this linear combination are linearly independent. Hence, this linear combination is zero only if each $\pmb{L}_k$, where $k \in \cup_{i \in S}R_i$, appears at least twice in the expansion $\sum_{i \in S}\sum_{k \in R_i} z_i\alpha_{i,k} \pmb{L}_k$, i.e., only if each $k \in \cup_{i \in S}R_i$ is contained in at least two distinct sets $R_i$ and $R_j$, $i \neq j$ and $i,j \in S$. Then a simple counting argument leads to the statement of this lemma.
\end{proof}

The following result can be used to manipulate the bound in Lemma~\ref{lem:fitting_support} to derive explicit lower bounds on locality.

For any $S \subseteq [N]$, let $G_S$ denote the subgraph of $G$ induced by the vertices in $S$, i.e., the vertex set of $G$ is $S$ and edge set is $\{(i,j) \in \mathcal{E}\vert i,j \in S\}$. The subgraph $G_S$ is the side information graph of the index coding problem obtained by restricting the index coding problem $G$ to the messages $x_i$, $i \in S$.

\begin{lemma} \label{lem:lower_bound_on_union}
Let $\pmb{L}$ be a valid scalar linear index code for $G$ with receiver queries $R_1,\dots,R_N$. For any $S \subseteq [N]$, we have
\begin{equation*}
 \left| \bigcup_{i \in S} R_i \right| \geq \minrank(G_S).
\end{equation*} 
\end{lemma}
\begin{proof}
The submatrix $\pmb{L}_S$ of $\pmb{L}$ consisting of the rows indexed by $S$ is a valid encoder matrix for the index coding problem $G_S$. Since, the receivers \mbox{$i \in S$} query only the coded symbols with indices \mbox{$k \in \cup_{i \in S}R_i$}, the submatrix of $\pmb{L}_S$ consisting of the columns with indices in $\cup_{i \in S}R_i$ is also a valid scalar linear encoder for $G_S$. Hence, the codelength $|\cup_{i \in S}R_i|$ of this index code is lower bounded by $\minrank(G_S)$.
\end{proof}

The next result follows immediately from Lemmas~\ref{lem:fitting_support} and~\ref{lem:lower_bound_on_union}.

\begin{corollary} \label{cor:lower_bound_r}
If $\pmb{L}$ is an optimal scalar linear encoder for $G$, i.e., has codelength equal to $\minrank(G)$, $\pmb{A}$ is a fitting matrix corresponding to $\pmb{L}$ and $\pmb{z} \in \nullsp(\pmb{A}) \setminus \{\pmb{0}\}$, then
\begin{equation*}
\sum_{i \in S} r_i ~\geq~ 2 \,\minrank(G_S),
\end{equation*} 
where $S=\supp(\pmb{z})$.
\end{corollary}
\begin{proof}
The matrix $\pmb{L}$ has linearly independent columns since the number of columns $\ell$ of $\pmb{L}$ satisfies
\begin{equation*}
\ell = \minrank(G) \leq \rank(\pmb{A}) \leq \rank(\pmb{L}) \leq \ell.
\end{equation*} 
The corollary holds since $r_i=|R_i|$ for scalar linear codes and the vectors $\pmb{L}_k$, $k \in \cup_{i \in S}R_i$ satisfy the conditions of Lemma~\ref{lem:fitting_support}.
\end{proof}

\section{Vector Linear Coding for Directed Cycles} \label{sec:directed_cycles}

\begin{figure}[!t]
\centering
\includegraphics[width=2in]{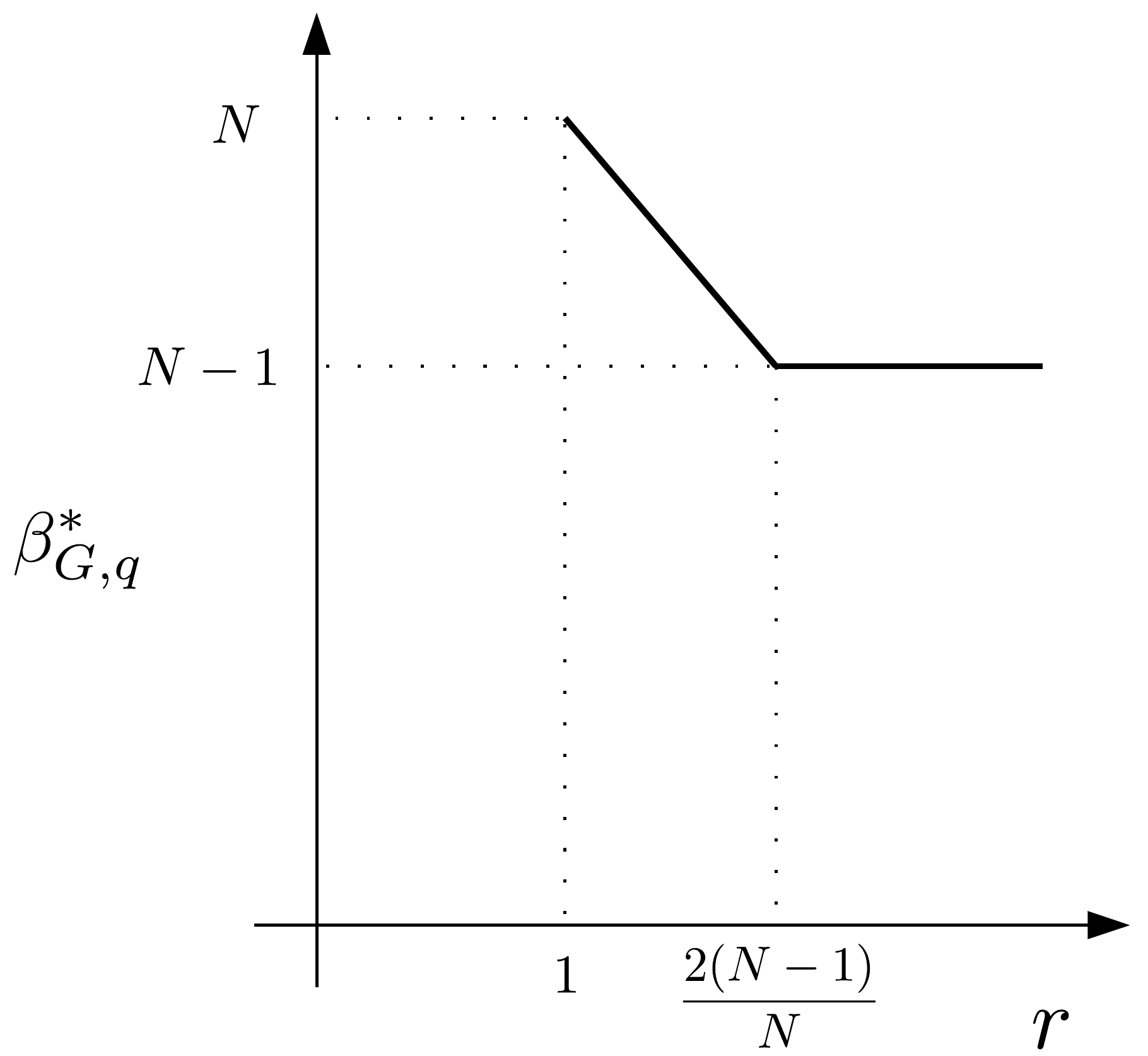}
\caption{The locality-rate trade-off of linear index codes for the directed $N$-cycle index coding problem.}
\label{fig:N_cycle_tradeoff}
\end{figure} 

In this section, we will consider the index coding problem where $G$ is a directed $N$-cycle, i.e., for $i=1,\dots,N-1$, $K_i = \{i+1\}$ and $K_N=\{1\}$. For $N \geq 3$, and over any finite field $\Fb_q$, we will show that
\begin{equation} \label{eq:N_cycle_tradeoff}
\beta^*_{G,q}(r) = \max\left\{N-1, \frac{N(N-1-r)}{N-2}\right\}, ~~~r \geq 1.
\end{equation} 
Note that locality $r \geq 1$ for any valid index coding scheme, and hence, $\beta^*_{G,q}(r)$ is defined for $r \geq 1$ only.
The trade-off between rate and locality is shown in Fig.~\ref{fig:N_cycle_tradeoff}. 
Sections~\ref{sec:converse_directed_cycle} and~\ref{sec:achievability_directed_cycle} provide the proofs for the converse and achievability, respectively, of this rate-locality trade-off.
The smallest locality at which the rate $N-1$, which is the minimum possible rate, is achievable is $r=2(N-1)/N$. This locality is achievable if the message length $M$ of the vector linear code is chosen carefully.
In Section~\ref{sec:locality_and_message_length}, we provide a detailed analysis of the effect of the message length $M$ on the locality $r$ when the broadcast rate is $N-1$.

\subsection{Converse} \label{sec:converse_directed_cycle}

In this subsection we will show that $\beta^*_{G,q}(r)$ is lower bounded by both $N-1$ and $N(N-1-r)/(N-2)$. It is clear that $\beta^*_{G,q}(r) \geq N-1$, since even without any locality constraints the smallest possible broadcast rate for the directed $N$-cycle is $N-1$. To complete the converse, we only need to show that $\beta^*_{G,q}(r) \geq N(N-1-r)/(N-2)$.

Suppose that the encoding matrix $\pmb{L}$ is valid with respect to a set of queries $R_1,\dots,R_N$. From Lemma~\ref{lem:single_queries}, the number of codeword symbols queried exactly once
\begin{equation*}
 |\Sc| = \sum_{i \in [N]}|\Sc_i| \geq M(2\beta-N\ravg).
\end{equation*} 
Hence there exists an $i \in [N]$ such that $|\Sc_i| \geq M(2\beta -N\ravg)/N$. Since $G$ is a directed cycle, without loss of generality, let us assume that 
\begin{equation} \label{eq:cycle_converse:2}
|\Sc_N| \geq M(2\beta - N\ravg)/N. 
\end{equation} 
We now relate this lower bound on $|\Sc_N|$ to the rank of $\pmb{L}$ to complete the converse.

For $i=1,\dots,N-1$, we have $\Kb_i=\{iM+1,iM+2,\dots,(i+1)M\}$ and $\Db_i = \{(i-1)M+1,\dots,iM\}$. From Theorem~\ref{thm:design_criterion}, for each $j \in \Db_i$ there exists a $\pmb{u}_j \lhd \Kb_i$ such that $\pmb{e}_j + \pmb{u}_j \in \colsp(\pmb{L})$, where $\colsp(\pmb{L})$ denotes the column span of the matrix $\pmb{L}$. Considering the first $N-1$ receivers $i=1,\dots,N-1$ and each of their demands $j \in \Db_i$, we obtain $M(N-1)$ such vectors $\pmb{e}_j + \pmb{u}_j$, all which lie in $\colsp(\pmb{L})$. Arranging these vectors into a matrix of size $MN \times M(N-1)$ we arrive at
\begin{equation} \label{eq:cycle_converse:1}
          \begin{bmatrix} 
          \pmb{I} & \pmb{0} & \cdots & \pmb{0} & \pmb{0} \\
          \pmb{C}_1 & \pmb{I} & \cdots & \pmb{0} & \pmb{0} \\
          \pmb{0} & \pmb{C}_2 & \cdots & \pmb{0} & \pmb{0} \\
          \vdots & \vdots &            & \vdots & \vdots  \\
          \pmb{0} & \pmb{0} & \cdots & \pmb{I} & \pmb{0} \\
          \pmb{0} & \pmb{0} & \cdots & \pmb{C}_{N-2} & \pmb{I} \\
          \pmb{0} & \pmb{0} & \cdots & \pmb{0} & \pmb{C}_{N-1}
          \end{bmatrix},
\end{equation} 
where each of the submatrices is of size $M \times M$. Note that the columns of this matrix are linearly independent.

Now considering $\rx_N$, we note that $\Db_N=\{(N-1)M+1,\dots,MN\}$. The coded symbols $\pmb{x}^\tr\pmb{L}_k$, $k \in \Sc_N$, are queried only by $\rx_N$. From~\eqref{eq:zeroes}, we deduce that $\pmb{L}_k \lhd \Db_N$ for all $k \in \Sc_N$, i.e., $\supp(\pmb{L}_k) \subseteq \Db_N$. Since the set of vectors $\{\pmb{L}_k \, | \, k \in R_N\}$ is linearly independent and $\Sc_N \subseteq R_N$, we observe that the vectors $\pmb{L}_k$, $k \in \Sc_N$, are linearly independent as well. Note that each of these vectors is a column of $\pmb{L}$ and hence lies in $\colsp(\pmb{L})$. Appending these $|\Sc_N|$ vectors as columns to the matrix in~\eqref{eq:cycle_converse:1}, we arrive at the block matrix
\begin{equation*}
 \pmb{A} = \begin{bmatrix} 
          \pmb{I} & \pmb{0} & \cdots & \pmb{0} & \pmb{0}  & \pmb{0} \\
          \pmb{C}_1 & \pmb{I} & \cdots & \pmb{0} & \pmb{0} & \pmb{0}  \\
          \pmb{0} & \pmb{C}_2 & \cdots & \pmb{0} & \pmb{0} & \pmb{0}  \\
          \vdots & \vdots &            & \vdots & \vdots & \pmb{0}   \\
          \pmb{0} & \pmb{0} & \cdots & \pmb{I} & \pmb{0} & \pmb{0}  \\
          \pmb{0} & \pmb{0} & \cdots & \pmb{C}_{N-2} & \pmb{I} & \pmb{0}  \\
          \pmb{0} & \pmb{0} & \cdots & \pmb{0} & \pmb{C}_{N-1} & \pmb{B}_N
          \end{bmatrix},
\end{equation*} 
where $\pmb{B}_N$ is an $M \times |\Sc_N|$ matrix with linearly independent columns. 
Note that each column of $\pmb{A}$ lies in $\colsp(\pmb{L})$, i.e., $\colsp(\pmb{A}) \subseteq \colsp(\pmb{L})$, and the columns of $\pmb{A}$ are linearly independent, i.e., $\rank(\pmb{A}) = M(N-1) + |\Sc_N|$.
Thus, we have
\begin{align*}
\ell &\geq \rank(\pmb{L}) \geq \rank(\pmb{A}) ~~~~~~~~~~~ (\text{since } \colsp(\pmb{L}) \supseteq \colsp(\pmb{A})) \\
&= M(N-1) + |\Sc_N| \\
&\geq M(N-1) + M(2\beta - N\ravg)/N ~~~~~~~(\text{using~\eqref{eq:cycle_converse:2}})
\end{align*} 
Using the fact that the broadcast rate \mbox{$\beta=\ell/M$}, the above inequality yields \mbox{$\beta \geq (N-1) + (2\beta - N\ravg)/N$}, which upon manipulation results in
\begin{equation} \label{eq:ravg_beta_bound_N_cycle}
\beta \geq {N(N-1-\ravg)}/{(N-2)}.
\end{equation} 
Since $\ravg \leq r$, we arrive at
$\beta \geq \frac{N(N-1-r)}{N-2}$.

\subsection{Achievability} \label{sec:achievability_directed_cycle}

In this subsection we show that the trade-off in~\eqref{eq:N_cycle_tradeoff} is achievable using linear index codes. We will show that the points $(r,\beta)=(1,N)$ and $(2(N-1)/N,N-1)$ are achievable. Then any point on the line segment $\beta = N(N-1-r)/(N-2)$, $1 \leq r \leq 2(N-1)/N$ can be achieved using time sharing between these two schemes.
The achievability of the points $\beta=N-1$ and $r>2(N-1)/N$ will follow immediately since the rate $N-1$ is already achievable with $r=2(N-1)/N$.

\subsubsection{Achieving $r=1$, $\beta=N$}

The point $(r,\beta)=(1,N)$ can be achieved trivially using the uncoded scheme, i.e., the transmitted codeword equals the message vector $\pmb{c}=\pmb{x} \in \Fb_q^{MN}$. Each receiver $\rx_i$ queries $\pmb{c}_{\Db_i}=\pmb{x}_{\Db_i}$ to meet its demand. The side information available at the receivers is not utilized by this scheme. Since the code length $\ell = MN$, we have $\beta=N$ and since $|R_i|=|\Db_i|=M$, we have $r=1$. Also note that this is a linear index code corresponding to the encoding matrix $\pmb{L}=\pmb{I}$.

\subsubsection{Achieving $r=2(N-1)/N$, $\beta=N-1$}

Example~\ref{ex:simple_code_cycle} provides a family of $N$ scalar linear codes for $G$, one for each choice of $i\in [N]$, with rate $\beta=N-1$. The $i^\tth$ code provides localities $r_i=r_{i+1}=1$ and $r_j=2$ for all $j \neq i,i+1$, where we interpret $i+1$ as $1$ if $i=N$. Using $\pmb{r}=(r_1,r_2,\dots,r_N)$ to represent the tuple of receiver localities, we observe that rate $N-1$ can be achieved with the following values of $\pmb{r}$
\begin{align}
\pmb{r}_1&=(1,1,2,2,\dots,2),~ \pmb{r}_2=(2,1,1,2,\dots,2),\dots, \nonumber \\
 \pmb{r}_{N-1}&=(2,2,\dots,2,1,1),~ \pmb{r}_N=(1,2,\dots,2,1). \label{eq:r_tuples}
\end{align}  

If $N$ is an odd integer, we time share the $N$ scalar linear codes corresponding to $\pmb{r}_1,\pmb{r}_2,\dots,\pmb{r}_N$. Observe that the overall scheme is a vector linear code for message length $M=N$, rate $N-1$ and locality $r=\ravg=2(N-1)/N$.

If $N$ is an even integer, we time share $N/2$ scalar linear codes corresponding to $\pmb{r}_1,\pmb{r}_3,\dots,\pmb{r}_{N-1}$, that yields a vector linear code with $M=N/2$, rate $N-1$ and $r=\ravg=2(N-1)/N$.

\subsection{Dependence of locality on message length} \label{sec:locality_and_message_length}

From the achievability scheme in Section~\ref{sec:achievability_directed_cycle} we observed that for $\beta=N-1$, the locality $r=\ravg=2(N-1)/N$ can be achieved using message length $M=N$ if $N$ is odd, and $M=N/2$ if $N$ is even.

\begin{lemma} \label{lem:minimum_M_for_locality}
Let $N \geq 3$. The message length $M$ of any index code that achieves locality $2(N-1)/N$ for the directed $N$-cycle satisfies $M \geq N$ if $N$ is odd, and $M \geq N/2$ if $N$ is even.
\end{lemma}
\begin{proof}
Consider any valid coding scheme with locality $r=2(N-1)/N$. There exists an $i \in [N]$ such that
$\frac{2(N-1)}{N} = r = r_i = \frac{|R_i|}{M}$,
that is 
\begin{equation*}
 M = \frac{|R_i|\,N}{2(N-1)}.
\end{equation*} 
If $N$ is odd, $N$ and $2(N-1)$ have no common factors, and since $M$ is an integer, we deduce that $M$ must be a multiple of $N$, i.e., $M \geq N$. If $N$ is even, using the fact $N/2$ and $N-1$ have no common factors we arrive at $M \geq N/2$.
\end{proof}

From Lemma~\ref{lem:minimum_M_for_locality}, it is clear that the minimum $M$ required to attain $r=2(N-1)/N$ at rate $N-1$ is $M=N$ if $N$ is odd and $M=N/2$ if $N$ is even. We will now derive the optimal locality when the message length is smaller than this quantity, i.e., $M<N$. We do so by analysing the two cases, $M < N/2$ and $N/2 \leq M < N$.

\subsubsection{Locality when $M < N/2$}

From~\eqref{eq:ravg_beta_bound_N_cycle} we deduce that for any vector linear scheme of rate $\beta=N-1$, we have 
\begin{equation*}
 \ravg \geq 2(N-1)/N.
\end{equation*} 
Thus, $\sum_{i=1}^{N}|R_i| = MN\ravg \geq 2M(N-1)$. It follows that there exists an $i \in [N]$ such that 
\begin{equation} \label{eq:size_of_Ri_M_and_N}
|R_i| \geq \frac{2M(N-1)}{N} = 2M - \frac{M}{N/2}.
\end{equation} 
If $M<N/2$, considering the fact that $|R_i|$ is an integer, we deduce that $|R_i| \geq 2M$. Hence, $r_i=|R_i|/M \geq 2$, and thus, $r \geq 2$. This lower bound on $r$ can be achieved by simply using the scalar linear code of Example~\ref{ex:simple_code_cycle} $M$ times, leading to a vector linear code for message length $M$, rate $N-1$ and $r=2$. Note that this code still achieves the optimal value of average locality $\ravg=2(N-1)/N$.

\subsubsection{Locality when $N/2 \leq M < N$}

If $N$ is even, the message length $M=N/2$ is sufficient to attain $r=2(N-1)/N$. Thus it is enough to consider larger values of $M$, i.e., $N/2 \leq M < N$ only for $N$ odd.
From~\eqref{eq:size_of_Ri_M_and_N} and using the fact that $|R_i|$ is an integer, we arrive at $|R_i| \geq 2M-1$. Thus, 
\begin{equation*}
r \geq r_i \geq 2 - \frac{1}{M}.
\end{equation*} 
Assuming $N$ is odd, this lower bound on $r$ is achieved by time sharing the $M$ scalar linear codes from Section~\ref{sec:achievability_directed_cycle} corresponding to the tuples of localities 
\begin{equation*}
\pmb{r}_1,\pmb{r}_3,\dots,\pmb{r}_{N-3},\pmb{r}_N,\pmb{r}_2,\pmb{r}_4,\dots,\pmb{r}_{2M-(N+1)},
\end{equation*} 
see~\eqref{eq:r_tuples}. It is straightforward to show that this scheme has rate $N-1$, $r=2\,-\,1/M$ and $\ravg=2(N-1)/N$. 
Note that in the interval $N/2 \leq M < N$, the value of the optimal locality increases with $M$. Hence, the choice $M=(N+1)/2$ yields the smallest locality in this interval.


\section{Scalar Linear Coding when Minrank is $N-1$} \label{sec:minrank_N_minus_1}

We now characterize the optimal localities $r$ and $\ravg$ among scalar linear index codes for index coding problems $G$ with $\minrank(G)=N-1$. 
The message length $M=1$ for scalar codes, and hence, the receiver localities $r_i$ and the rate $\beta$ are integers. 
Since the minimum scalar coding rate is equal to $\minrank$, we are interested in the operating points corresponding to $\beta=N$ and $\beta=N-1$. In the rest of this section we will assume that $\minrank(G)=N-1$.

The side information graph $G$ contains at least one directed cycle. Otherwise, $G$ is a directed acyclic graph and its minrank is equal to $N$~\cite{YBJK_IEEE_IT_11}, a contradiction. Let $N_c$ denote the length of the smallest directed cycle contained in $G$. 

If $N_c=2$, there exist $i,j \in [N]$ such that $(i,j),(j,i) \in \mathcal{E}$, i.e., $i \in K_j$ and $j \in K_i$. The following scalar linear code attains the minimum possible locality $r=\ravg=1$ and the minimum possible rate $\beta=N-1$ simultaneously. 
Transmit $x_i+x_j$ followed by transmitting the remaining $N-2$ information symbols uncoded. $\rx_i$ and $\rx_j$ can decode using $x_i+x_j$, and the remaining receivers query their demands directly from the codeword.

In the rest of this section we will assume that $N_c \geq 3$. Observe that rate $\beta=N$ can be achieved with smallest possible localities $r=\ravg=1$ using uncoded transmission. 

We will now consider the case $\beta=\minrank(G)=N-1$.
Let $\pmb{L}$ be any scalar encoder matrix with codelength $\ell=N-1$, and $\pmb{A}$ be a corresponding fitting matrix. Since $\ell=\minrank(G)$, we have
$\ell \leq \rank(\pmb{A}) \leq \rank(\pmb{L}) \leq \ell$,
and hence, $\rank(\pmb{A})=\rank(\pmb{L})=N-1=\ell$. Thus, the nullspace of $\pmb{A}$ contains a non-zero vector.

\begin{lemma} \label{lem:GS_contains_cycle}
If $\pmb{z} \in \nullsp(\pmb{A}) \setminus \{\pmb{0}\}$ and $S=\supp(\pmb{z})$, the subgraph $G_S$ of $G$ induced by the vertices $S$ contains at least one directed cycle.
\end{lemma}
\begin{proof}
Observe that the $|S| \times |S|$ submatrix $\pmb{A}'$ of $\pmb{A}$ composed of the rows and columns of $\pmb{A}$ with indices in $S$ fits $G_S$. Since $\pmb{Az}=\pmb{0}$, the columns of $\pmb{A}$ indexed by $S$ are linearly dependent. This implies that the columns of the submatrix $\pmb{A}'$ are linearly dependent as well, and hence, $\rank(\pmb{A}') \leq |S|-1$. It follows that
$\minrank(G_S) \leq \rank(\pmb{A}') \leq |S|-1$.
Clearly $G_S$ is not a directed acyclic graph since otherwise $\minrank(G_S)=|S|$.
\end{proof}

In order to use Corollary~\ref{cor:lower_bound_r}, we now derive a lower bound on $\minrank(G_S)$.

\begin{lemma} \label{lem:lower_bound_minrank_GS}
If $\minrank(G)=N-1$, then for any $S \subseteq [N]$, $\minrank(G_S) \geq |S|-1$.
\end{lemma}
\begin{proof}
Consider the following valid scalar linear code for $G$. Encode the information symbols $x_i$, $i \in S$, using the optimal scalar linear code for $G_S$, and use uncoded transmission for the remaining symbols. The length of this code is lower bounded by $\minrank(G)=N-1$, hence we obtain
$\minrank(G_S) + N - |S| \geq N - 1$.
\end{proof}

We now prove the main result of this section.

\begin{theorem}
If $\minrank(G)=N-1$ and the smallest directed cycle in $G$ is of length $N_c \geq 3$, the optimal locality for scalar linear coding for $G$ with rate $N-1$ is
\begin{equation*}
 r = 2 \text{ and } \ravg = \frac{N+N_c-2}{N}.
\end{equation*} 
\end{theorem}
\begin{proof}
\emph{Converse:} Let $\pmb{A}$ be a fitting matrix corresponding to any valid scalar linear code for $G$ with rate $N-1$. Let \mbox{$\pmb{z} \in \mathcal{N}(\pmb{A}) \setminus \{\pmb{0}\}$} and $S=\supp(\pmb{z})$.
From Corollary~\ref{cor:lower_bound_r} and Lemma~\ref{lem:lower_bound_minrank_GS}, $\sum_{i \in S}r_i \geq 2\,\minrank(G_S) \geq 2(|S|-1)$. Using the trivial bound $r_i \geq 1$ for $i \notin S$, we have
\begin{align*}
\sum_{i \in [N]} r_i &= \sum_{i \in S} r_i + \sum_{i \notin S}r_i \\
&\geq 2(|S|-1) + N-|S| = N + |S| - 2. 
\end{align*} 
From Lemma~\ref{lem:GS_contains_cycle}, we know that $G_S$ contains a cycle, and hence, the number of vertices $|S|$ in $G_S$ is at least $N_c$. Thus,
\begin{equation*}
\ravg = \frac{\sum_{i \in [N]}r_i}{N} \geq \frac{N +|S|-2}{N} \geq \frac{N+N_c-2}{N}.
\end{equation*} 
Since $N_c \geq 3$, we have $r \geq \ravg \geq (N+1)/N$, and since $r$ is an integer we conclude that $r \geq 2$.

\emph{Achievability:} Let $C \subseteq [N]$ be the set of vertices that form the smallest directed cycle in $G$. Note that the subgraph $G_C$ is a directed cycle of length $|C|=N_c$. We encode the symbols $x_i$, $i \in C$, using the scalar linear code given in Example~\ref{ex:simple_code_cycle} and send the remaining $N-N_c$ symbols uncoded. This achieves the codelength $N-1$. From Example~\ref{ex:simple_code_cycle}, the sum locality within the cycle $\sum_{i \in C}r_i = 2(N_c-1)$ and the maximum locality within the cycle $\max_{i \in C}r_i=2$. The locality of the remaining receivers is $r_i=1$, $i \notin C$. This scheme achieves the optimal values of $r$ and $\ravg$ for rate $N-1$.
\end{proof}

\begin{remark}
Corollary~V.2 of~\cite{KSCF_arXiv_18} shows that, over the binary field $q=2$, if $\minrank(G)=N-1$, then a scalar linear coding rate of $N-1$ is achievable for any choice of locality $r \geq 2$. 
In contrast, our results show that $r=2$ is optimal (if $N_c \geq 3$) and also provide the optimal value of the average locality $\ravg$ for scalar linear codes over an arbitrary finite field $\Fb_q$.
\end{remark}


\appendices

\section{Proof of Theorem~\ref{thm:zeroes}} \label{app:proof:thm:zeroes}

We will design a new encoding matrix $\pmb{L}'$ by modifying the subset of the columns of the given matrix $\pmb{L}$ corresponding to the column indices $\Sc_1 \cup \cdots \cup \Sc_N$. 
For the remaining indices $k \in \Mc_1 \cup \cdots \cup \Mc_N$, the $k^\tth$ columns of $\pmb{L}$ and $\pmb{L}'$ are equal, i.e., $\pmb{L}_k = \pmb{L}'_k$. For an arbitrary $i \in [N]$, we will now explain the construction of the column vectors $\pmb{L}'_k$, $k \in \Sc_i$. Since the symbols $\pmb{x}^\tr\pmb{L}'_k$, $k \in \Sc_i$, are queried only by $\rx_i$ and are unused by other receivers, we only need to consider the constraints that are imposed by the demands of $\rx_i$ while designing the column vectors $\pmb{L}'_k$, $k \in \Sc_i$.

We will introduce the notation which will be used in the rest of the proof. For any $E \subseteq [MN]$, let $U_{E} = \spp(\pmb{e}_k, k \in E)$, i.e., $U_E$ is the subspace of all vectors whose support is a subset of $E$. 
For any $F \subseteq [\ell]$, let $V_F = \spp(\pmb{L}_k, k \in F)$ and $V'_F = \spp(\pmb{L}'_k, k \in F)$.
Since $R_i = \Sc_i \cup \Mc_i$, we have $V_{R_i} = V_{\Sc_i} + V_{\Mc_i}$ and $V'_{R_i} = V'_{\Sc_i} + V'_{\Mc_i}$, where the addition corresponds to sum of subspaces.
From Theorem~\ref{thm:design_criterion} and using the fact that $\pmb{L}$ is a valid encoder matrix, we have $\pmb{e}_j \in V_{R_i} + U_{\Kb_i}$, for all $j \in \Db_i$, i.e., we have 
\begin{equation} \label{eq:lem:zeroes}
U_{\Db_i} = \spp(\pmb{e}_j, j \in \Db_i) \subseteq V_{R_i} + U_{\Kb_i} = V_{\Sc_i} + V_{\Mc_i} + U_{\Kb_i}.
\end{equation} 
Again using Theorem~\ref{thm:design_criterion}, we observe that $\pmb{L}'$ allows $\rx_i$ to decode its demand if and only if 
\begin{equation} \label{eq:lem:zeroes:2}
U_{\Db_i} \subseteq V'_{R_i} + U_{\Kb_i} = V'_{\Sc_i} + V'_{\Mc_i} + U_{\Kb_i}.
\end{equation} 

Using the validity of the encoder matrix $\pmb{L}$, we will first lower bound $|\Sc_i|$ which is the number of coded symbols queried uniquely by $\rx_i$.
From~\eqref{eq:lem:zeroes}, we obtain
\begin{align} \label{eq:lem:zeroes:3}
U_{\Db_i} = (V_{\Sc_i} + V_{\Mc_i} + U_{\Kb_i}) \cap U_{\Db_i}.
\end{align} 
Let $\Vin = (V_{\Mc_i} + U_{\Kb_i}) \cap U_{\Db_i}$ denote the subspace of $V_{\Mc_i} + U_{\Kb_i}$ contained in $U_{\Db_i}$, and let $\Vout$ be any subspace such that $\Vout \cap U_{\Db_i} = \{\pmb{0}\}$ and $\Vin + \Vout = V_{\Mc_i} + U_{\Kb_i}$. Continuing from~\eqref{eq:lem:zeroes:3}, we claim that
\begin{align}
U_{\Db_i} &= (V_{\Sc_i} + \Vin + \Vout) \cap U_{\Db_i} \label{eq:lem:zeroes:4} \\
&= \left( (V_{\Sc_i} + \Vout) \cap U_{\Db_i} \right) + \Vin. \label{eq:lem:zeroes:5}
\end{align} 
It is clear that the subspace in~\eqref{eq:lem:zeroes:4} contains the subspace in~\eqref{eq:lem:zeroes:5} since $\Vin \subseteq U_{\Db_i}$. To prove that~\eqref{eq:lem:zeroes:4} is contained in~\eqref{eq:lem:zeroes:5}, assume that $\pmb{y} \in V_{\Sc_i} + \Vout$ and $\pmb{z} \in \Vin$ are such that $\pmb{y} + \pmb{z} \in U_{\Db_i}$. Since $\pmb{z} \in \Vin \subseteq U_{\Db_i}$ and $\pmb{y} + \pmb{z} \in U_{\Db_i}$, we conclude that $\pmb{y} \in U_{\Db_i}$ as well. Thus $\pmb{y} \in (V_{\Sc_i} + \Vout) \cap U_{\Db_i}$, and hence, $\pmb{y} + \pmb{z} \in \left((V_{\Sc_i} + \Vout) \cap U_{\Db_i}\right) + \Vin$.

Considering the dimensions of the subspaces in~\eqref{eq:lem:zeroes:5}, we have 
\begin{equation} \label{eq:lem:zeroes:6}
\dim(U_{\Db_i}) \leq \dim((V_{\Sc_i} + \Vout) \cap U_{\Db_i}) + \dim(\Vin).
\end{equation} 
In order to proceed with the proof of the theorem, we will now show that $\dim((V_{\Sc_i} + \Vout)\cap U_{\Db_i}) \leq |\Sc_i|$. To do so, assume that $\pmb{y}_j + \pmb{z}_j$, $j=1,\dots,n$, form a basis for $(V_{\Sc_i} + \Vout)\cap U_{\Db_i}$ where $\pmb{y}_j \in V_{\Sc_i}$ and $\pmb{z}_j \in \Vout$ and $n=\dim((V_{\Sc_i} + \Vout)\cap U_{\Db_i})$. 
If the scalars $\alpha_1,\dots,\alpha_n$ are such that $\sum_j \alpha_j \pmb{y}_j = \pmb{0}$, then
\begin{equation*}
 \sum_{j = 1}^{n} \alpha_j(\pmb{y}_j + \pmb{z}_j) = \sum_{j=1}^{n} \alpha_j \pmb{z}_j \in \Vout.
\end{equation*} 
Since $\pmb{y}_j + \pmb{z}_j \in U_{\Db_i}$, we also observe that $\sum_j \alpha_j(\pmb{y}_j + \pmb{z}_j) \in U_{\Db_i}$. Using the fact $\Vout \cap U_{\Db_i} = \{\pmb{0}\}$, we deduce that $\sum_j \alpha_j(\pmb{y}_j + \pmb{z}_j)=\pmb{0}$, and hence, $\alpha_1=\cdots=\alpha_n=0$. We conclude that $\pmb{y}_1,\dots,\pmb{y}_n \in V_{\Sc_i}$ are linearly independent, and therefore
\begin{equation} \label{eq:lem:zeroes:7}
\dim((V_{\Sc_i} + \Vout) \cap U_{\Db_i}) = n \leq \dim(V_{\Sc_i}) \leq |\Sc_i|.
\end{equation} 

Note that the columns of $\pmb{L}$ and $\pmb{L}'$ corresponding to the column indices $\Mc_i$ are equal, and hence $V_{\Mc_i}'=V_{\Mc_i}$. Using this fact together with~\eqref{eq:lem:zeroes:6} and~\eqref{eq:lem:zeroes:7}, we have
\begin{align*}
|\Sc_i| &\geq \dim((V_{\Sc_i} + \Vout)\cap U_{\Db_i}) \\
        &\geq \dim(U_{\Db_i}) - \dim(\Vin) \\
        &= \dim(U_{\Db_i}) - \dim((V_{\Mc_i} + U_{\Kb_i}) \cap U_{\Db_i}) \\
        &= \dim(U_{\Db_i}) - \dim((V'_{\Mc_i} + U_{\Kb_i}) \cap U_{\Db_i}).
\end{align*} 

From~\eqref{eq:lem:zeroes:2}, in order to satisfy the claim of this theorem, it is sufficient to chose the $|\Sc_i|$ vectors $\pmb{L}'_k$, $k \in \Sc_i$, such that $\pmb{L}'_k \lhd \Db_i$, i.e., $\pmb{L}'_k \in U_{\Db_i}$ and 
\begin{equation*}
 V'_{\Sc_i} \,+\, \left( (V'_{\Mc_i} + U_{\Kb_i}) \cap U_{\Db_i} \right) \supseteq U_{\Db_i}.
\end{equation*} 
This is always possible since the difference in the dimensions of $U_{\Db_i}$ and $(V'_{\Mc_i} + U_{\Kb_i}) \cap U_{\Db_i}$ is at the most $|\Sc_i|$. 
One way to construct $\pmb{L}'_k$, $k \in \Sc_i$, is as follows. We begin with a basis for $(V'_{\Mc_i} + U_{\Kb_i}) \cap U_{\Db_i}$. These vectors form a linearly independent set in $U_{\Db_i}$. We extend this set to a basis for $U_{\Db_i}$. The number of additional vectors in this basis is $\dim(U_{\Db_i}) - \dim((V'_{\Mc_i} + U_{\Kb_i}) \cap U_{\Db_i})$. These additional vectors together with $|\Sc_i|-\dim(U_{\Db_i}) + \dim((V'_{\Mc_i} + U_{\Kb_i}) \cap U_{\Db_i})$ all-zero vectors are chosen as the columns $\pmb{L}'_k$, $k \in \Sc_i$.



\end{document}